\documentclass{jgcc} 


\usepackage{hyperref}
\theoremstyle{plain} 

\pdfoutput=1
\usepackage{lastpage}
\jgccdoi{14}{2}{3}{10150} 
\jgccheading{}{\pageref{LastPage}}{}{}{Oct.~14,~2022}{Feb.~16,~2023}{}   

\begin{document}

\title[Multi-recipient and threshold encryption]{Multi-recipient and threshold encryption\\ based on hidden multipliers}
\author[Vitaly~Roman'kov]{Vitaly~Roman'kov} 	
\address{Sobolev Institute of Mathematics (Omsk Branch), Omsk, Pevtsova~13, 644099, Russia}	
\email{romankov48@mail.ru}  
\thanks{The research was supported in accordance with the state task of the IM SB RAS, project FWNF-2022-003.}	

\begin{abstract}
  \noindent Let $S$ be a pool of $s$ parties and Alice be the dealer. In this paper, we propose a scheme that allows the dealer to encrypt messages in such a way that only one authorized coalition of parties (which the dealer chooses depending on the message) can decrypt. At the setup stage, each of the parties involved in the process receives an individual key from the dealer. To decrypt information, an authorized coalition of parties must work together to use their keys. Based on this scheme, we propose a threshold encryption scheme.   For a given message  $f$ the dealer can choose any threshold $m = m(f).$  More precisely, any set of parties of size at least $m$ can evaluate  $f$; any set of size less than $m$ cannot do this. 
Similarly, the distribution of keys among the included parties can be done in such a way that authorized coalitions of parties will be given the opportunity to put a collective digital signature on any documents. This primitive can be generalized to the dynamic setting, where any user can dynamically join the pool $S$. In this case the new user receives a key from the dealer. Also any user can leave the pool $S$. In both cases, already distributed keys of other users do not change. The main feature of the proposed schemes is that for a given $s$ the keys  are distributed once and can be used multiple times. This property distinguishes the proposed schemes from the most of such schemes known in the literature. However, it should be noted that similar schemes have already been proposed (see, for example, the schemes  by M. Bellare, A. Boldyreva, K. Kurosawa, J. Staddon (2007) and schemes by C. Delerabl$\acute{\rm e}$e and D. Pointcheval (2008)). 

The proposed scheme is based on the idea of hidden multipliers in encryption. As a platform, one can use both multiplicative groups of finite fields and groups of invertible elements of commutative rings, in particular, multiplicative groups of residue rings. We propose two versions of this scheme.
\end{abstract}

\maketitle

\section*{Introduction}\label{sec:Intro}
In the early days of cryptography, most schemes were designed for a single-sender/single-receiver scenario. Currently, there are scenarios where many recipients (or many senders) need to share power to use the cryptosystem. The main motivation behind multi-recipient and threshold cryptography has been to develop methods for working with single-sender/multi-receiver scripts. For these concepts and related definitions of secrecy, see, for example, \cite{BBKS}.

In multi-sender cryptography, the cryptosystem protects information by encrypting it and distributing it among multiple parties. Information in the form of a message is transmitted to the parties in encrypted using the public key and the corresponding private key is shared between the parties involved. To decrypt information, an authorized coalition of parties must cooperate to use their keys. With a threshold cryptosystem, in order to decrypt an encrypted message or sign a message, multiple parties (greater than a certain threshold number) must cooperate in a decryption or signature protocol.  Sharing secrets was introduced in 1979 by Shamir \cite{Shamir} and Blakely \cite{Bl}.
Since then, many applications have emerged for several different types of cryptographic protocols. The basics of threshold cryptography are contained in the fundamental works  \cite{BBH,CG,SDFY,Des1,DF,Des2,F,SG}. See also survey \cite{Bei}  and papers  \cite{RD,BeiDes,CSB,CGMW,DKL,PZ,BEN,VES,OK,KGH,SS,Sti,HC} for some secret sharing schemes.  

For example,  a large organisation is carrying out a complex project that involves various groups of its employees. The organization distributes individual keys among the parties, allowing various groups of them to receive the necessary information about the project. Each such group has the right to receive only a certain part of the full information. Moreover, such a system can provide for the possibility of collective digital signatures on the reports of various groups of parties. This requires a certain system of access for various groups of parties to various pieces of information about the project, which can be constructed using threshold cryptography methods. To achieve this, the keys are distributed in such a way that each coalition of  parties can get a signature from their  keys without disclosing any information about their  keys. This example shows that the considered schemes is a natural primitive. 

The main goal of this paper is to build schemes that provide for the reuse of once distributed secret resources. This is possible only in cases where the allocated private keys are not revealed when the message $f$   is  decrypted on their basis. It should also be possible to add, remove or replace qualified group members without changing their  keys.  These properties show the advantages of the proposed scheme in comparison with the many known  such schemes.

It should be noted that there are schemes with multiple use of the initially distributed  keys. For example, in \cite{CGMW}, the authors propose  a threshold secret sharing scheme based on polynomial interpolation and the Diffie-Hellman problem. In this scheme shares can be used many times for the reconstruction of multiple secrets. This scheme involves the use of hash functions and has a number of other significant differences from the schemes proposed in this paper.  C. Delerabl$\acute{\rm e}$e and D. Pointcheval \cite{DP} proposed  a generalization of threshold public-key encryption  to the dynamic setting,
where any user can dynamically join the system, as a possible recipient; the
sender can dynamically choose the authorized set of recipients, for each ciphertext, the sender can dynamically set the threshold $m$ for decryption capability
among the authorized set. 

We propose  schemes such that the initial distribution of  keys between all participants in the process is carried out either using a secure communication channel, or using the protocol of secret key transfer over an open communication channel. The entire further process is carried out over an public network.

We consider $s$ parties.   
So, we propose new multi-recipient and threshold $(m, s)$-schemes that allow qualified parties to receive  message  $f$. For a given message $f$, the dealer can choose any qualified set of recipients and any threshold $m=m(f)$  without  redistributing  keys. One of the two versions of the  encryption scheme is monotonic and the other is not. The proposed threshold scheme is monotonic, i.e., any set of $k\geq m$ participants is qualified. We also offer two versions of the collective digital signature associated with the proposed versions of multi-receive encryption, respectively.

\bigskip 
 Notation:  $\mathbb{Z}$ -- set of integer numbers, 
 $\mathbb{Z}_n = \mathbb{Z}/n\mathbb{Z}$ -- residue ring,  $\mathbb{N}$ -- set of nonnegative integer numbers, $\mathbb{N}_k = \{1, \ldots , k\}.$ For an element $g$ of some group, $|g|$ denotes its order.

\section{Construction of fields and residue rings with prescribed orders of subgroups of multiplicative groups}\label{sec:1}

The main idea behind the corresponding algorithm  is the following statement similar to  
\cite{Hand}, Fact 4.59 (\cite{Rintro}, Theorem 38, Proposition 39).
\label{pro:1}
\begin{prop} 
\label{pro:1}{\it   Let $b\geq 3$ be an odd integer, and suppose that 
$b=1+rq$, where $q$ is an odd prime and $r$ is an even positive integer. 
\begin{enumerate}
\item
If there exists an integer $a$ satisfying $a^{b-1}\equiv 1 (\bmod\,b)$  and gcd($a^{r}-1, b$) = $1$, 
then for any prime divisor $p$ of the number $b$, $p \equiv  1(\bmod\, 2q)$, in particular, $p \geq  2q + 1.$ 
\item If additional
 the condition $r\leq 4q+2$ is satisfied, then $b$ is prime.
\item If $b$ is prime, the probability that a randomly selected base $a, 1 \leq a \leq b-1$, 
satisfies $a^{b-1}\equiv 1 (\bmod\,b)$  and gcd($a^{r}-1, b$) = $1$ is $\frac{q-1}{q}$.
\end{enumerate}}
\end{prop}
\begin{proof}
 Let $p$ be a prime divisor of $b$. Then condition (1) implies  $a^{b-1} \equiv 1(\bmod\, p).$ We also have $a^r \not= 1(\bmod\, p).$ On the other hand, $a^{p-1} \equiv 1(\bmod\, p)$ by Fermat's Little Theorem. Then in the group $\mathbb{Z}_p^{\ast}$ we get
 $a^r\not= 1, a^{b-1}=1, a^{p-1} =1$. It follows from the inequality $a^r\not= 1$ together with $(a^r)^q=1$ that $|a^{r}|=q.$ By Lagrange Theorem 
 $|a|\vdots q$. Then $p-1\vdots q$. Since $p-1$ is even $p-1\vdots 2q$  and so $p\equiv 1(\bmod\, 2q)$. The statement (1) is proved. 
 Let (2) is satisfied. Suppose that $b$ is composite. The $b$ is divisible by at least two primes
 $p_1$ and $p_2$. By (1) $p_1, p_2 \geq 2q+1.$   
  Lets do the calculations:
$$(2q+1)^2 = 4q^2 +4q +1 \leq b = 1 + rq  \leq (4q+2)q = 4q^2 + 2q +  1.$$
This inequality obviously false. The statement (2) is proved. 
For prime $b$ the first condition $a^{b-1} \equiv 1(\bmod\, b)$ is true by Fermat's Little Theorem and the second condition  gcd$(a^r-1 , b)=1$ is satisfied if and only if
 $a^r \neq 1 (\bmod\, b)$. In the field $ \mathbb{F}_b$, the equation $x^r = 1 $
has at most $r$ roots, one of which is equal to $1$ and the other is  $-1$. Therefore, on the interval $ 1 < a < b-1 $, there are at most $ r-2 $
numbers $ r $ for which $ a^r = 1$ in the field $ \mathbb{F}_b.$ This means
that the probability of choosing such $a$ is no more than $\frac{r-2}{b-3}
\sim \frac{r}{qr} = \frac {1}{q}$. Thus (3) is proved. 
\end{proof}
The following algorithm   recursively generates an odd prime $b$, and then chooses random  integers $r, q < r$, until $b = 1 + rq$  can be proven prime using  for some base $a$. By
proposition \ref{pro:1} the probability of such bases is $\sim 1 - 1/q$  for prime $b$. On the other hand, if $b$ is composite, then most bases $a$ will fail to satisfy the condition $a^{b-1}\equiv 1 (\bmod\,b)$. 
Let's describe  this algorithm. 
\begin{enumerate}
\item Select a random  odd integer $b = 1 + rq$, where $r$ is an even number. 
\item We start with an  odd prime $q = q_1$.
\item Let's choose a number $r$ at random:
$$
q + 1 \leq r \leq 4q + 2.
$$
\item Consider
$$
b = 1 + rq,\, 
q\leq r \leq 4q+2.
$$
\item
 Choose randomly the number $a = a_1$ within $1 <
 a < b-1$ and check the fulfillment of conditions  from
proposition \ref{pro:1}. If $a = a_1$ does not satisfy these conditions, then we take
another random number $a = a_2$. So we repeat a sufficient number of times:
$a = a_1, a_2, \ldots, a_k$ until we find a suitable value
$a$.

If you succeed in doing this, then $b$ is prime. We put $q = q_2 = b$ and
repeat the construction starting from first step. We do this until
we get a big enough prime.
\end{enumerate}
If, with a large number of trials for $a$, it was not possible to execute the
conditions  of proposition \ref{pro:1}, then we change $r$ and repeat everything again.

Suppose that the constructed number $b$ is indeed
prime. Then the probability of finding the number
$a$ with the given properties  from proposition \ref{pro:1} 
 is  $\sim \frac {1}{q}$.
 
Note also that the so constructed prime $b$ will be
 greater than $q^2$ because $q \leq r$ and $b = 1 + rq $. The primes $q_1, q_2, \ldots $ obtained as a result of this sequential construction grow no less than quadratically.

Let's ask a question: how realistic is it to find
a prime number $b = 1 + rq $ under the indicated constraints $ q \leq r \leq
4q + 2$, choosing an even $r$.

First of all, note that, by the famous Dirichlet theorem, the progression
$n = 2qt + 1$ ($t = 0, 1, 2, 3, \ldots $)  contains infinitely many
prime numbers. We are interested in primes $n$ of the indicated form with
possible small parameters $t = 1, 2, \ldots $. If the
generalized Riemann hypothesis is true, then the smallest prime number in the indicated
sequence does not exceed $c(\varepsilon) q^{2+ \varepsilon}$ for any $\varepsilon > 0$ 
($c(\varepsilon $ is a constant,
depending on $\varepsilon $). Numerical experiments show that
primes in the specified sequence occur quite
often and close to its beginning. Note also that, according to the theory known
numbers to Cramer's hypothesis $p_{n + 1} - p_n = O(\ln^2p_n)$ (here $p_n$
denotes the $n$th prime number in order). A similar conclusion follows from the generalized Riemann hypothesis.

Suppose we need to construct a prime $p$ such that $p-1=r$ and $r$ is divisible by the product of $s+1$ pairwise coprime numbers $d$ and $t = \prod_ {i = 1}^st_i,$ i.e., $r = dtr'$.  This can be effectively done by the process just described, by choosing the parameter $r$  that is divisible by $dt$ and $r'$ that is divisible by  $q$.  Then we obtain the prime number $p = 1 + r$ and build the finite field $\mathbb{F}_p$ of order $p.$  We can assume that $d = t - 1$ which  is coprime with any number $t_i$. The order $p-1$ of the multiplicative group $\mathbb{F}_p^{\ast}$ is divisible by  $r$. 

Therefore $\mathbb{F}_p^{\ast}$ contains $s$ cyclic subroups $T_i$ = gp($u_i$),
 where $|u_i| = t_i, i = 1, \ldots , s,$ and a subgroup $F$ = gp($f_0$) of order $d$. Let $g$ generates $\mathbb{F}_p^{\ast}$. The elements $u_i$ are efficiently computable by the formula $u_j = g^{\frac{r}{t_j}}$.  The element $f$ is computed as $f_0 = g^\frac{r}{d}.$
 
 Of course, there is another way to find the prime number $p$ for which $p-1$ is divisible by the product $dt $, as above. We select the even numbers $r'$ in a certain interval and check the simplicity of the number $ p = 1 + drr' $ using well-known tests, for example, the Miller-Rabin test (see \cite {Hand}). The check goes on until a simple $p$ is obtained. This method is effective and often used in practical cryptography. 
 
The indicated method of constructing the subgroups $T_i, i = 1, \ldots , s$ and $F$, as above, is obviously extended to residue rings, in particular, to rings of the form $\mathbb{Z}_n, n = pq$, where $p$ and $q$ are different primes. In this case, we can construct the primes $p$ and $q$ with the desired sets of divisors for the numbers $p-1$ and $q-1$, and then use them in our construction. More precisely, suppose we want to construct a ring $ \mathbb {Z} _n $ such that the multiplicative group $\mathbb{Z}_n^{\ast}$ of order 
$\varphi (n) = (p-1)(q -1)$ has subgroups $W_1, \ldots, W_k $ of orders $r_1 = t_1s_1, \ldots, r_k = t_ks_k $, respectively, where any pair $ (t_i, s_i) $ consists of two coprime numbers. Then we choose a prime number $p$ such that $p-1$ is divisible by $ t = \prod_{i = 1}^kt_i $, and $ q $  such that $q-1$ is divisible by 
$s = \prod_{i = 1}^ks_i.$
For each  $i = 1, \ldots, k $ let the element $u_i \in \mathbb{Z}_p$ is of order $t_i$, and similarly 
$v_i \in \mathbb{Z}_q $ is of order $s_i.$ By the Chinese Remainder Theorem, from the system of equations 
$$ \begin{cases} w_i \equiv u_i (\bmod \, p), \\ w_i \equiv v_i (\bmod\, q).\end{cases} $$ 
 \noindent we find $w_i$. Obviously, $w_i$ is of order $r_i$ and we can define $W_i = $ gp($w_i$) for $i = 1, \ldots , k.$
 
  Using residue rings instead of fields as platforms for encryption can have its benefits. The schemes proposed in this paper are based on the difficult solvability of calculating the order of an element of a multiplicative group of a field or a residue ring, respectively. In the field $\mathbb{F}_q, q = p^r$, with a known primary decomposition of the number $q-1$ (the order of the multiplicative group $\mathbb{F}_q^{\ast}$), there exists a polynomial algorithm for calculating the order of an arbitrary element $g\in \mathbb{F}_q^{\ast}$. See \cite{Hand}, algorithm 4.79, page 162. The specified primary decomposition makes the proposed schemes vulnerable to the case of a finite field. In particular, schemes are  vulnerable when quantum computers are used to generate such decompositions. In cryptography, when using residue rings as platforms, it is assumed that the order of the ring's multiplicative group is unknown. However, quantum computing in this case also makes the circuit vulnerable.

\section{General organization}
\label{sec:3}

In this section, we give a formal organization of the process. Consider a dealer Alice and an initial pool of $s$ participants $A_1,\ldots , A_s$. At each step, this pool can change. For brevity we keep denotion $s$ in future descriptions.  Alice estimates the possible number $s_{new}$ of new participants in the process. Let $s_{max} = s+s_{new}$. Then Alice chooses a platform: a finite field or a residue ring of sufficiently large size, in which $s_{max}+1$  subgroups $C_1, \ldots , C_{s_{max}}$ and $F$ of sufficiently large pairwise coprime orders can be distinguished. The corresponding process is described in the previous section. Possible decryption keys are the orders of the first $s_{max}$ subgroups $C_1, \ldots , C_{s_{max}}$.  The last subgroup $F$ will serve as the message space, i.e., each message will be encoded by its element. The corresponding process should be described in a special way. For finite fields and residue rings such processes are well known.

Next, Alice randomly selects a subset of $s$   subgroups among $C_1, \ldots , C_{s_{max}}$ and distributes their orders (keys) among the participants. The participant $A_i$ thus receives the private key $t_i$. In the future, new participants may appear who, upon registration, receive their private keys, chosen by Alice from the previously unused orders of cyclic subgroups $C_1, \ldots , C_{s_{max}}$. The keys of the retired participants are not used in the future. Therefore, this probabilistic distribution scheme allows Alice to generate and distribute  individual keys  among the parties.

Individual keys $t_i$ are transmitted by Alice in encrypted form, regardless of the model used. For such transmission, Alice opens the encryption system. The system can be either symmetric or public key. All participants in the process must have complete necessary information about the encryption system. Various protocols can be used to transfer keys $t_i$, for example, the Diffie-Hellman protocol.

 After this stage, the parties and dealer can communicate following one of the next two Communication models. In the first model (private channel model), the parties communicate through a complete synchronous network of secure and reliable point-to-point channels. Any set of parties has access to the messages sent to the parties in the set. In the second model (broadcast channel model),  the parties communicate through a public  channel. A set of parties can obtain all the messages circulating between the parties. Alice publishes the encrypted messages $f$ in the network she uses. A more secure method using a trusted server is as follows.
Let the scheme is not monotonic. The server opens a separate room where it invites every participant from the authorized set. Each of them receives a password to enter and the ability to operate with the received message. On entry, participant $A_i$ modifies the message using its key $t_i$. A member outside of that coalition cannot do so. Participants can log in with nicknames to hide who owns the key if it is somehow calculated by other participants. It is possible that the present participants are visible only to a trusted server. Note that in the versions of the protocol proposed below, the key is calculated as a discrete logarithm.

More secret is the scheme in which members of the coalition pass their keys to a trusted server upon entry. The server performs the corresponding operations without declaring intermediate results. Only the final result is announced. With such an organization, the coalition members do not have data to calculate the keys of other coalition members. They may only attempt to compute the shared key of the coalition. Therefore, it must ensure the security of the protocol. If the coalition is small, the dealer can use the keys of virtual participants prepared in advance by him, formally including them in the coalition. 
It is assumed that they are known to the server, which will perform their operations on its own. 

 In case the scheme is monotonous or threshold, the trusted server first gathers a plurality of participants in a separate location, also using one-time passwords. It then checks to see if the set of participants gathered is capable of deciphering the message. After that, it allows them to carry out their operations or  carries out these operations, as described above, announcing only the final result.
 
 A coalition signature is carried out in a similar way. 

\section{Multi-recipient encryption  protocol}
\label{sec:4}
  
   The main idea used to construct a new multi-recipient  protocol is the encryption scheme proposed in the works of the author \cite {RomRSA} and  \cite {RomRSAPDM}. Suppose Alice installs the following cryptographic system, whose platform $K$ is either the multiplicative group $\mathbb{F}_p^{\ast}$ of a finite field $\mathbb{F}_p$, where $p$ is a prime number,  or the multiplicative group $\mathbb{Z}_n^{\ast}$ of a residue ring $\mathbb{Z}_n, n = pq $, where $ p $ and $ q $ are distinct primes. In the field case, the parameter $ p $ is public. In the  residue ring case, the parameters $ p, q $ are private, and $ n $ is public. 
We denote by $K$ the multiplicative group $\mathbb{F}_p^{\ast}$ or $\mathbb{Z}_n^{\ast}$ respectively.   
   
 The scheme works as follows.   
   
   Alice chooses two subgroups $F$ and $H$ of $K$ of coprime exponents: exp($F$) $= k$ and exp($H$) $= l$,  by the method described in the previous section.  Recall that, the exponent of a group is defined as the least common multiple of the orders of all elements of the group.  The subgroup $F$ serves as the message space, and $H$ is the space of hidden multipliers. Both of these subgroups $F$ and $H$ are publicly available. Then $k$ and $l$ are Alice's private  numbers. She  also computes a private number  $l'$ such that $ll'=1 (\bmod\, k).$ It follows, that $f^{ll'}=f$ for any element 
$f\in F.$ 
Suppose Bob wants to send a message  to Alice. Alice will receive and decrypt this message. The algorithm works  as follows:
\begin{enumerate}
\item Bob encodes the message as $f\in F$, chooses $h\in H$ at random and sends $c = hf$ to Alice.  
\item Alice computes $$c^{ll'} = (h^l)^{l'}f^{ll'} = f. $$
\end{enumerate}

\medskip
{\bf Secrecy.}
The secrecy of the proposed scheme is based on the intractability of calculating the order of an element in a finite field or in a residue ring of the considered type.  Note that the ability to calculate the orders of elements in the  residue ring $\mathbb{Z}_n$ makes it possible to reveal transmitted messages without knowing the decryption key. Consider, for example, the RSA  system with standard notation for its elements ($\mathbb{Z}_n, e,$ and so on). Indeed, the order of the encrypted message $c = m^e$ in the multiplicative group 
$\mathbb{Z}_n^{\ast}$ is equal to the order of the original message $m$, since the degree $e$  (the encryption key) is relatively prime to the order of the multiplicative group $\mathbb{Z}_n^{\ast}.$ Let's say the attacker calculated this order of $t$.  Then he can find a one-time decryption key $d_m$  from the equality $ed_m=1 (\bmod \, t)$ and calculate 
$c^{d_m} = m (\bmod\, n)$. This scheme has a number of advantages over the standard RSA. First, the encryption uses an easier-to-perform multiplication operation, rather than exponentiation. Secondly, different keys are used, which provides different types of ciphertexts for the same message. This gives the semantic secrecy property.

\begin{rem}
Of course, one can use as a platform the multiplicative group $K^{\ast}$ of any commutative associative ring $K$ with unity, provided that large subgroups of coprime exponents can be chosen in $K$, and the problem of calculating the order of an element is intractable. 
One of the advantages of this system over the original RSA version is its semantic secrecy. See \cite{RomRSA}, \cite{RomRSAPDM}  or \cite{RomCryptology} for details.  Using residue rings instead of fields as platforms for encryption can have its benefits. The just described scheme is based on the difficult solvability of calculating the order of an element of a multiplicative group of a field or a residue ring, respectively. In the field $\mathbb{F}_q, q = p^r$, with a known primary decomposition of the number $q-1$ (the order of the multiplicative group $\mathbb{F}_q^{\ast}$), there exists a polynomial algorithm for calculating the order of an arbitrary element $g\in \mathbb{F}_q^{\ast}$. See \cite{Hand}, algorithm 4.79, page 162. 
\end{rem}

Let ${\mathcal S}$ be the system  which is  organized and managed by Alice.  Let 
 $ \{A_1, \ldots, A_s\} $ is the set of users in ${\mathcal S}$ at the considered step.
 
 \bigskip
 {\bf Version 1} 

\medskip
When setting up the system $ {\mathcal S}$, Alice takes a set of pairwise coprime positive integers $t_1, \ldots , t_{s_{max}}$.

Let $t = \prod_{i=1}^st_i$. Alice also defines $d = t -1$ or $d= t +1.$  Then Alice chooses as a  platform the group $K$  of large order $r$  where $r = dtr',$ for some $r'\in \mathbb{N}$, while simultaneously defining the set of subgroups 
$C_i =$ gp($u_i$) ($i = 1,  \ldots , s$) and $F$ of $K$ of the   orders $t_1, \ldots , t_s$ and $d$, respectively. This group can be chosen as a subgroup of the multiplicative group of a finite field or residue ring according to the process described in  Section \ref{sec:2}. The group  $K$ is public. The subgroups $C_1, \ldots , C_s, F$ and their corresponding orders $t_1, \ldots , t_s, d$ are private. Let $H=\prod_{i=1}^sC_i$. The subgroup $F$ serves as the space of possible  messages $f$, and $H$ is the space of hidden multipliers.

Then Alice distributes the numbers (keys) $t_i$ among the current users $A_1, \ldots , A_s$ of  $ {\mathcal S}$. For simplicity me assume that these keys are distributed at such a way, that each participant $A_i$ gets $t_i$ for $i = 1, \ldots , s.$  These keys are for future reuse. The remaining unused keys are stored for distribution to new users of the system, if any appear in the future. This distribution is carried out either over a secure communication channel, or is transmitted in encrypted form over an open channel. These keys are for future reuse.  

 Let $ f, f \in F $ be a message that Alice wants to send to some (qualified) set of users of the system $S(f) = \{A_ {i_1}, \ldots , A_{i_w}| 1 \leq i_1 < \ldots < i_{w} \leq s\}.$ Alice acts as follows:
\begin{enumerate}
\item Alice randomly selects nontrivial elements $v_{i_j} \in C_{i_j},\, j = 1, \ldots , w.$ Then she computes  $c = \prod_{j=1}^wv_{i_j}\cdot f^{t/\prod_{j=1}^wt_{i_j}}$.    She sends $c$ to the coalition $S(f).$
\item Members of the coalition $ S(f) $ sequentially raise the obtained element $c$ to the power  $t_j$ for  $j = 1, \ldots, w.$ In the case $d=t-1$, they get element
$$f^t = f^{d+1} = f.$$ 
If $d = t+1$, they get
$$f^t = f^{d-1} = f^{-1}$$
\noindent and compute $f.$ 
\end{enumerate}

Obviously, Alice can send the message $f$ to any possible coalition of users in this way. Any unqualified coalition will not be able to reveal the message $f$ in some natural way. If this coalition does not contain the user $A_{i_j}$ it cannot remove the factor $v_{i_j}$.

This scheme is not monotonous. Moreover, interference with the disclosure of a secret by any member outside the qualified coalition results in an incorrect secret.

\bigskip
{\bf Version 2.} 

\medskip 
This version can be used to decrease  $d$. 

When setting up the system $ {\mathcal S}$, Alice takes a set of pairwise coprime positive integers $t_1, \ldots , t_{s_{max}}$ and $d$ such that $t_i = 1 (mod\, d)$ for all $i.$  Such a set of numbers $t_i$ exists for any $d$ by the famous Dirichlet theorem, according to which there are infinitely many such primes $t_i$. 

Now Alice takes as above the multiplicative  group $K$  of a finite field or residue ring of  order  $r = dt_{max}r', t_{max} = \prod_{i=1}^{n_{max}}t_i,$  while simultaneously defining the set of subgroups 
$C_i =$ gp($u_i$) ($i = 1,  \ldots , s_{max}$) and $F$ of the  group $K$ of  orders $t_1, \ldots , t_{s_{max}}$ and $d$, respectively. The corresponding algorithm has been given in Section \ref{sec:2}. The group $K$ is public. The subgroups $C_1, \ldots , C_{s_{max}}, F$ and their corresponding orders $t_1, \ldots , t_{s_{max}}, d$ are private. Let $H=\prod_{i=1}^{s_{max}}C_i$. The subgroup $F$ serves as the space of messages $f$, and $H$ is the space of hidden multipliers. 

For simplicity we assume that the numbers $t_i$  are distributed as keys among the current users of  ${\mathcal S}$ at such a way, that each participant $A_i$ gets $t_i$ for $i = 1, \ldots , s.$ Let's denote  $t = \prod_{i=1}^st_i.$ These keys are for future reuse. The remaining unused keys are stored for distribution to new users of the system, if any appear in the future. This distribution is carried out either over a secure communication channel, or is transmitted in encrypted form over an open channel. 

 Let $ f\in F $ be a message that Alice wants to send to some (qualified) set of users of the system $S(f) = \{A_ {i_1}, \ldots ,  A_{i_w} | 1 \leq i_1 < \ldots < i_{w} \leq s\}.$ Alice acts as follows:

\begin{enumerate}
\item Alice randomly selects nontrivial elements $v_{i_j} \in C_{i_j},\, j = 1, \ldots , w.$ Then she computes  $c = \prod_{j=1}^wv_{i_j}\cdot f$.    She sends $c$ to the coalition $S(f).$
\item Members of the coalition $ S(f) $ sequentially raise the  element $c$  and elements successively received from it to the power $t_j$ for  $j = 1, \ldots, w.$ 
\end{enumerate}

The first step gives the element 
$$c_{i_1} = v_{i_2}^{t_{i_1}} \cdots v_{i_w}^{t_{i_1}}\cdot f.$$
This means that the first factor has been removed from the record. The rest of the factors before $f$ retained their orders, since these orders are coprime to $t_{i_1}$.
Continuing the process, they sequentially remove all factors except $f$, which remains unchanged for all exponentiations.
As a result, they get the element $f$.
 
Unlike version 1, this version is monotonous. Any coalition containing a qualified coalition also reveals the secret. This is due to the fact that each raising to the power  does not change the factor $f$. 

\bigskip
{\bf Coalition signature.}

The following descriptions use the above notations.
As usual, the process of setting up and verifying a signature is in a certain sense the opposite of the process of setting up and recovering  a message by coalition. In this case, there is a dealer, say, Alice, who organizes the process, a set of parties $A_1, \ldots , A_s$, and a set of possible verifiers, which for simplicity we consider to be singleton $D$. As before, Alice creates an auxiliary cryptographic system. Using this system, Alice distributes among the parties  not the numbers $t_1, \ldots , t_s$, as described in the above schemes, but the generating elements $u_1, \ldots , u_s$ of  subgroups $C_1, \ldots , C_s$  relatively. Any element $f$ of the subgroup $F$ defined as in the schemes 1 and 2 can be considered as a document to sign by an authorized coalition  $\{A_{i_1}, \ldots , A_{i_k}\}$.  

Let $f\in F$  be a signature document for some coalition $\{A_{i_1}, \ldots , A_{i_k}\}$.    The signing procedure consists in the fact that each $A_{i_j}$ first selects an element $a_{i_j}$ from the subgroup gp($u_{i_j}$) distributed to him. Then they successively multiply $f$ by the selected elements, resulting in a signed document $f_{sign}=a_{i_1} \ldots a_{i_k} \cdot f$.

Alice also gives  the numbers $t_1, \ldots , t_s$ to the verifier $D$. To check the correctness of the signature of the given coalition on the document $f$, $D$ calculates $t =t_{i_1} \ldots t_{i_k}$ and then raises $f_{sign}$ to the power $t$. In both schemes 1 and 2, with the correct statement of the signature, it should turn out to be $f$. The efficiency and security of this algorithm are similar to the corresponding qualities of the above schemes. 

It should be noted that if scheme 2 is used in this way, the verifier can verify the correctness of the signature of any sub-coalition. If this is not acceptable, appropriate additional steps should be taken.

\section{ ($m, n$)-threshold encryption scheme }
\label{sec:4}
This proposition bases on the version 2 of the multi-recipient encryption protocol described in Section \ref{sec:4}. 
Let us prove a preliminary statement.
\begin{lem}
\label{le:3}
Let $s\in \mathbb{N}.$  For any $k\in \mathbb{N}_s$, there exists a set $T(k) = 
\{ t_1, \ldots , t_{l_k}\}$, which can be represented as a union of  subsets $T_j(k),\, j = 1, \ldots , s,$ such that the union of any $k$ subsets coincides with $T(k)$, and the union of a smaller number is strictly less than $T(k)$.
\end{lem}
\begin{proof}
Induction by $k$. The statement is true for $k = 1$, when one can define $l_1 = 1$ and $T_1(j)= \{t_1\}$ for any $j$. Assume that the statement is true for $k-1$ and $T(k-1) =  \{t_1, \ldots , t_{l_{k-1}}\}.$ We enumerate all pairs of  subsets $T_j(k-1), j = 1, \ldots , s$ with distinct numbers as $V_1, \ldots , V_{\binom{s}{k-1}}$.  Then we take elements $t_{l_{k-1}+1}, \ldots , 
t_{l_{k-1}+ \binom{s}{k-1}}$ and include each $t_{l_{k-1}+i}$ into all $T_j(k-1)$ except for those contained in  $V_i$.  Therefore we can  set 
$T(k) = \{t_1, \ldots , t_{l_k}\}$, where $l_k =  1 + \binom{s}{1} + \ldots + 
\binom{s}{k-1}$, satisfying the required condition. 
\end{proof}
Note that for the indicated construction $ T(1) \subset T (2) \subset \ldots \subset T(s) $ and $|T (n)| = 2^s-1.$ Moreover,  for any $j$, the inclusions $ T_j(1) \subset T_j(2) \subset \ldots \subset T_j(s) $ are satisfied.

Let ${\mathcal S}$ be the system  which is  organized and managed by Alice.  Let 
 $ \{A_1, \ldots, A_s\} $ be the set of users in ${\mathcal S}.$ 
 
 \medskip
For  $b=2^s-1$, Alice takes a set of pairwise coprime positive integers $T = \{t_1, \ldots , t_b\} \cup \{d\}$.  Let $t = \prod_{i=1}^rt_i.$  Then Alice (following section 2) chooses a large simple finite field $\mathbb{F}_p,\, p - 1 = r,$, where $r = dtr',$ for some $r'\in \mathbb{N} $, simultaneously defining  subgroups
$W_i =$ gp($w_i$) ($i = 1, \ldots , b$) and $F$ of the multiplicative group $\mathbb{F}_p^{\ast}$ of  orders $t_1, \ldots , t_b $ and $d$ respectively. 

By Lemma \ref{le:3}, Alice defines a representation of the set $T$ in the form of a union of subsets $T(k)$
for $k = 1, \ldots , b.$  For each $j \in \mathbb{N}_n$, Alice computes $\bar{t}_j = \prod_{i\in T(j)}t_i$. Then she computes the keys $\tilde{t}_j = \bar{t}_j\bar{t}_j^{-}$, where $\bar{t}_j^{-} = 
\bar{t}_j^{-1} (\bmod\, d).$

Then Alice distributes the keys $\tilde{t}_j$ among the participants according to the indices.  Each participant $S_j$ receives the key $\tilde{t}_j$.  

We suppose that Alice wants to develop a  ($m, b$) threshold encryption  scheme. In each subgroup $W_i$, where $i \leq l_m$, Alice chooses a nontrivial element $g_i$. Then she computes $$c= \prod_{i \in T(m)}g_i\cdot f$$
\noindent 
and sends this element to all participants. 

Let $C$ be a coalition, consisting of $z \geq m$ members. Coalition members consistently raise $c$ to exponents equal to their keys. Since the product of all their keys is divisible by any value $t_i $ for $i \leq l_k $, the result is the message $f$. This does not happen if the coalition has fewer than $k$ members. Hence, it is a ($k, b$) secret sharing scheme. 

\bigskip
{\bf Properties and security}

\medskip
The semantic secrecy of the above schemes is based on the difficult solvability of the problem of calculating the exponent of an element in the platform under consideration (finite field or commutative associative ring with unity, in particular,  residue ring). Indeed, let there be two secrets $m_1$ and $m_2$, one of which is transmitted in the form $c = tm$, where $t$ and $m $ have coprime orders. If the attacker can calculate the exponents of the elements,  he will calculate $e_i = $exp($m_i$) and $e_i' = $ exp($c^{- 1}m_i$) for $i = 1, 2.$ Then he compares the sets of prime divisors for two pairs $ e_1, e_1 '$ and $ e_2, e_2' $. Only in the pair corresponding to the transmitted secret,  such a set for $ e_i '$ does not contain prime divisors of $e_i$.

If the problem of calculating the exponent of a protocol platform element is intractable, then this scheme is semantically secret. 

In the case of the field $\mathbb {F}_p$, to calculate the orders of the elements of the group $ \mathbb{F}_p^{\ast} $, it is sufficient to know the primary decomposition of the number $p-1 $ (see \cite{Hand}). The ability to solve the problem of calculating the order of an element of the group $\mathbb{Z}_n^{\ast}, n = pq$ ($p, q $ are different primes) gives an algorithm for calculating transmitted messages, that is, it solves the RSA problem (see \cite{RomRSA}). 

There is very little public data in the proposed schemes. In the case of the field 
$\mathbb{F}_p $, only its order $p$ is known, but the primary decomposition of the number $p-1$ is unknown. In the case of the residue ring $ \mathbb{Z}_n$, only the module $n$ is known.

\medskip

\section{Conclusions}

In this paper, we have studied two versions of a new scheme of multi-recipient encryption and threshold encryption  based on hidden multipliers in finite fields or commutative associative rings, in particular, residue rings.  We also propose a ($m, n$) - threshold scheme, where $ m $ is chosen by the dealer for each session without any additional allocations. The main feature of the proposed schemes is that the keys are distributed once among users  and can be used multiple times. Our schemes are secure against passive attacks and semantically secure under assumption that the problem of calculation the exponent of element of a protocol platform is intractable.  It is also easy to see that the version 2 of the scheme is monotonic, but the version 1 is not. We also offer two versions of the collective digital signature associated with the proposed versions of multi-receive encryption, respectively. All proposed schemes allow you to dynamically add and remove users of the main pool. Previously distributed keys are not changed.

\bibliographystyle{plain}
\bibliography{main_bibliography1}

 \end{document}